\documentclass[journal]{IEEEtran}

\usepackage{amsmath}
\usepackage{amsthm}
\usepackage{amssymb}
\usepackage{verbatim}
\usepackage{graphicx}
\usepackage{ifthen}
\usepackage{color}

\theoremstyle{plain}
\newtheorem{lemma}{Lemma}
\newtheorem{rem}{Remark}
\newtheorem{theorem}{Theorem}
\newtheorem{assumption}{Assumption}

\newtheorem{proposition}[theorem]{Proposition}
\newtheorem{definition}{Definition}
\newtheorem{corollary}{Corollary}
\theoremstyle{definition}

\newcommand{\argmin}{\operatornamewithlimits{argmin}}
\allowdisplaybreaks[4]

\usepackage{amsfonts}
\usepackage{times}
\usepackage{latexsym}
\usepackage{amssymb}
\usepackage{amsmath}
\usepackage{cite}
\usepackage{verbatim}
\newtheorem{prop}{Proposition}

\newcommand{\R}{\mathbb{R}}
\newcommand{\Z}{\mathbb{Z}}
\newcommand{\ra}{\rightarrow} 
\newcommand{\ua}{\uparrow}

\newcommand{\Exp}[1]{\mathbb{E}\left[#1\right]} 

\def\bb0{{\mathbb{0}}}

\def\bb{{\mathbf{b}}}

\def\b0{{\mathbf{0}}}
\def\opt{\mathsf{OPT}}

\def\b1{{\mathbf{1}}}

\def\cS{\mathcal{S}}

\def\sfc{{\mathsf{c}}}

\def\sfr{{\mathsf{r}}}

\def\sf0{{\mathsf{0}}}

\def\nn{\nonumber}

\newboolean{showcomments}
\setboolean{showcomments}{true}
\newcommand{\jk}[1]{  \ifthenelse{\boolean{showcomments}}
{ \textcolor{red}{(JK says:  #1)}} {}  }
\newcommand{\todo}[1]{\ifthenelse{\boolean{showcomments}}
{ \textcolor{red}{(to do:  #1)}}{}}

\newcommand{\revision}[1]{#1}
\newcommand{\ignore}[1]{}

\begin{document}

\newlength{\figurewidth}\setlength{\figurewidth}{0.6\columnwidth}



\title{\fontsize{23}{23}\selectfont Multiple Server SRPT with Speed Scaling is Competitive
}

\newcounter{one}
\setcounter{one}{1}
\newcounter{two}
\setcounter{two}{2}

\author{Rahul Vaze and~Jayakrishnan Nair
\thanks{Rahul Vaze is with the School of Technology and Computer Science, Tata Institute of Fundamental Research, India. Jayakrishnan Nair is with the Department of Electrical Engineering, IIT Bombay, India.}}

\maketitle

\begin{abstract}
  Can the popular shortest remaining processing time (SRPT) algorithm
  achieve a constant competitive ratio on multiple servers when server
  speeds are adjustable (speed scaling) with respect to the flow time
  plus energy consumption metric? This question has remained open for
  a while, where a negative result in the absence of speed scaling is
  well known.  The main result of this paper is to show that
  multi-server SRPT \revision{with speed scaling} can be constant competitive, with a competitive
  ratio that only depends on the power-usage function of the servers,
  but not on the number of jobs/servers or the job sizes (unlike when
  speed scaling is not allowed).  When all job sizes are unity, we
  show that round-robin routing is optimal and can achieve the same
  competitive ratio as the best known algorithm for the single server
  problem. Finally, we show that a class of greedy dispatch policies,
  including policies that route to the least loaded or the shortest
  queue, do not admit a constant competitive ratio. When job arrivals
  are stochastic, with Poisson arrivals and i.i.d. job sizes, we show
  that random routing and a simple gated-static speed scaling
  algorithm achieves a constant competitive ratio.
\end{abstract}

\section{Introduction}

How to route and schedule jobs are two of the fundamental problems in
multi-processor/multi-server settings, e.g. microprocessors with
multiple cores. Microprocessors also have the flexibility of variable
speed of operation, called {\it speed scaling}, where to operate at
speed $s$, the power utilization is $P(s);$ typically,
$P(s) =s^\alpha$, with $2\le \alpha \le 3$.  Speed scaling is also
available in modern queuing systems where servers can operate at
variable service rates with an appropriate cost function $P(.)$.

Increasing the speed of the server reduces the response times
(completion minus arrival time) but incurs a larger energy cost. Thus,
there is a natural tradeoff between the the {\it flow time} (defined
as the sum of the response times across all jobs) and the total energy
cost, and a natural objective is to minimize a linear combination of
the flow time and total energy, called \emph{flow time plus energy}.

In this paper, we consider the {\it online} problem of routing,
scheduling, and speed scaling in a multi-server setting to minimize
the flow time plus energy, where jobs arrive (are released) over time
and decisions have to be made causally. On the arrival of a new job, a
centralized controller needs to make a causal decision about which
jobs to process on which server and at what speed, where preemption
and migration is allowed. By migration, we mean that a job can be
preempted on one server and restarted on another server later.  The
model, however, does not allow job splitting, i.e., a job can only be
processed on a single server at any time.


For this problem, both the stochastic and worst case analysis is of
interest. In the stochastic model, the input (job sizes and
arrival instants) is assumed to follow a distribution, and performance
guarantees in expectation are derived.  In the worst case analysis,
the input can be generated by an adversary, and the performance metric
is the competitive ratio, that is defined as the maximum of the ratio
of the cost of the online algorithm and the optimal offline algorithm
(that knows the entire input sequence ahead of time).


\subsection{Prior Work}
\subsubsection{Single Server}
For a single server, it is known that Shortest Remaining
Processing Time (SRPT) is an optimal scheduling policy, and the only
decision with speed scaling is the optimal dynamic speed choice. There
is a large body of work on speed scaling in the single server setting
\cite{yao1995scheduling, baruah1991line, irani2007algorithms,
  bansal2007speed, albers2007energy, wierman2009power,
  lam2008speed,Andrew2010} both in the stochastic as well as worst
case settings, where mostly $P(s) =s^\alpha$ is used, under various
assumptions, e.g. bounded speed $s\in [0, S]$
\cite{bansal2008scheduling}, with and without deadlines
\cite{chan2009optimizing, han2010deadline,cote2010energy}, etc.

In the stochastic model, \cite{wierman2009power} showed that a simple
fixed speed policy (called gated static) that depends only on the
load/utilization and is independent of the current number of
unfinished jobs/sizes has a constant multiplicative gap from the
`unknown' optimal policy. Further work in this direction can be found
in \cite{gebrehiwot2014optimal, gebrehiwot2016energy}, where
\cite{gebrehiwot2016energy} derived the mean response time under the
SRPT algorithm.  For the worst case, there are many results
\cite{yao1995scheduling, baruah1991line, irani2007algorithms,
  bansal2007speed, albers2007energy, lam2008speed,
  bansal2008scheduling, chan2009optimizing,
  han2010deadline,cote2010energy,bansal2009speedconf,Andrew2010}. A
key result in this space was proved in \cite{bansal2009speedconf},
where an SRPT-based speed scaling algorithm is proved to be
$(3+\epsilon)$-competitive algorithm for an arbitrary power function
$P(\cdot)$. In \cite{Andrew2010}, using essentially the same ideas as
in \cite{bansal2009speedconf}, but with a more careful analysis, a
slightly modified SRPT-based speed-scaling algorithm is shown to be
$(2+\epsilon)$-competitive algorithm, also for an arbitrary power
function.

In the worst case setting, when considering speed scaling, two classes
of problems are studied:~(i) unweighted and~(ii)~weighted, where
in~(i) the delay that each job experiences is given equal weight in
the flow time computation, while in~(ii) it is scaled by a weight that
can be arbitrary. The weighted setting is fundamentally harder that
the unweighted one, and it is known that constant-competitive online
algorithms are not possible \cite{bansal2009weighted} in the weighted setting, even for a
single server, while constant competitive algorithms are known for the
unweighted case, even for arbitrary energy functions, e.g., the $(2+
\epsilon)$-competitive algorithm proposed in \cite{Andrew2010}. To
circumvent the negative result for the weighted case, typically, the
online algorithm is allowed a speed augmentation of $1+\delta$
compared to the optimal offline algorithm, in which case algorithms
with $O(1)$ competitive ratios are possible, where $O(1)$ depends
on~$\delta$.

\subsubsection{Multiple Servers}
With multiple servers without speed scaling (when the server speeds
are fixed), to minimize just the flow time, a well known {\it
  negative} result from \cite{leonardi2007approximating} states that
the SRPT algorithm (which always processes the $m$ smallest jobs with
$m$ servers) that requires both preemption and job migration has a
competitive ratio that grows as the logarithm of the ratio of the
largest and the smallest job size and the logarithm of the ratio of
the number of jobs and the number of servers.  Moreover,
\cite{leonardi2007approximating} also showed that no online algorithm
can do better than SRPT when server speeds are fixed.

With multiple servers, one critical aspect is whether job migration is
allowed or not. With job migration, a preempted job can be processed
by any of the servers and not necessarily by the server where it was
partially processed first.  Remarkably in
\cite{awerbuch2002minimizing}, a non-migratory algorithm that only
requires preemption is proposed that achieves the same competitive
ratio as SRPT.
A more positive result for SRPT is that if it is allowed a speed
augmentation of $(2-1/m)$ (respectively, $(1+\delta)$) over the
offline optimal algorithm, then it has a constant competitive ratio of
$1$ (respectively, a constant constant depending on $\delta$); see
\cite{phillips1997optimal,fox2011online}.

The problem to minimize flow time and energy with speed scaling in the presence of multiple servers
 has been studied in
\cite{lam2008competitive, greiner2009bell, gupta2010scalably,
  lam2012improved, gupta2012scheduling, Devanur18} under the worst case scenario. The homogenous server case
was studied in \cite{lam2012improved, greiner2009bell,Devanur18}, i.e., $P(.)$
is identical for all servers, while the heterogenous case was
addressed in \cite{gupta2010scalably, gupta2012scheduling}, where
$P(.)$ is allowed to be different for different servers.

For the unweighted flow time and energy problem under the homogenous
server case, a variant of the round robin algorithm without migration
has been shown to have a competitive ratio of $O(1)$
\cite{lam2008competitive} with $(1+\delta)$ augmentation with bounded
server speeds. This result was extended in \cite{greiner2009bell} for
the weighted flow time plus energy using a randomized server selection
algorithm that also does not use migration.

For the heterogenous server setting with $(1+\delta)$ augmentation,
\cite{gupta2010scalably, gupta2012scheduling}, derived algorithms that
assigns job to server that cause least increase in the projected
future weighted flow and a variant of processor sharing, respectively,
that are $O(1)$ competitive in unweighted and weighted flow time plus
energy. Moreover, if for server $k$, $P_k(s) = s^{\alpha_k}$, then the
algorithm in \cite{gupta2010scalably,Devanur18} has a competitive
ratio dependent on $\alpha$ without any need for speed augmentation.

In the stochastic setting, for multiple servers, the flow time plus
energy problem with multiple servers is studied under a fluid model
\cite{asghari2014energy, mukherjee2017optimal} or modelled as a Markov
decision process \cite{gebrehiwot2017near}, and near optimal policies
are derived.

%
%

Our focus in this work is on the unweighted flow time plus energy
under the homogenous server setting, where in the context of the prior
work we want to answer the following {\it open} questions: 
(i) \revision{under the
worst case setting, is it possible to achieve a constant competitive
ratio with algorithms simpler than the ones proposed \cite{gupta2010scalably, gupta2012scheduling}, that require  resource augmentation, or \cite{Devanur18} that is computationally expensive.}  In particular, can SRPT do so, since it is a
widely used and simple to implement algorithm?  This question is also
directly related to the limitation of SRPT without speed scaling as
shown in \cite{leonardi2007approximating}, and whether SRPT with and
without speed scaling are fundamentally different. (ii) For the
stochastic setting, can simple algorithms achieve near optimal
performance without the need of fluid limit approximations?

\subsection{Our Contributions}

Let the number of (homogenous) servers be $m.$

\begin{itemize}

\item The SRPT algorithm, with speed chosen as
  $P^{-1}\left(\frac{n}{m}\right)$ if $n \ge m$ and $P^{-1}(1)$ if
  $n < m$, where $n$ is the number of unfinished jobs, is shown to be
  $c$-competitive, where
  $$c = P(2-1/m) \left(2 + \frac{2}{P^{-1}(1)}
    \max(1,P(\bar{s}))\right),$$ \revision{where $\bar{s}$ is a
    constant associated with the function $P(\cdot).$} This means the
  above algorithm is constant-competitive, with a competitive ratio
  that is independent of the number of servers as well as the workload
  sequence.\footnote{While $m$ does appear in the expression for the
    competitive ratio, note that $2-1/m$ is trivially upper bounded by
    2.} This result is proved under mild regularity assumptions on the
  power function~$P(\cdot),$ which can be further relaxed using
  standard arguments \cite{bansal2009speedconf}. For the special case
  $P(s) = s^{\alpha},$ where $\alpha \in (1,2)$ , we derive another
  bound of $3 + \frac{2}{2-\alpha}$ on the competitive ratio; this
  bound is tighter than the previous one for $\alpha$ smaller than $2.$
  Similar to the algorithm proposed in \cite{gupta2010scalably, Devanur18}, the
  competitive ratio of the our SRPT-based policy also depends on
  $\alpha$. However, the algorithm proposed here is much simpler, and
  comes with a lower implementation complexity.
\item An important conclusion to draw from this result is that SRPT
  with speed scaling is fundamentally different as compared to the
  case when speed scaling is not allowed; in the latter setting, the
  competitive ratio depends on the number of jobs and their sizes
  \cite{leonardi2007approximating}. Thus, allowing for speed scaling,
  the ever popular SRPT is shown to be robust in the multiple server
  setting.
\item With speed scaling, we also derive some lower bounds for the
  immediate dispatch case when the job has to be assigned to a server
  instantaneously on its arrival and cannot be migrated across
  servers, though preemption within a server is allowed. Under this
  setting, we show that greedy routing policies, that assign a new job
  to the currently least loaded server or to the historically least
  loaded server have a competitive ratio of at least
  $\Omega(m^{1-1/\alpha})$ \revision{no matter how the speed is chosen}. Moreover, even when immediate dispatch is
  not necessary (i.e., jobs can wait in a common queue), but job
  migration across servers is not allowed, we show that the
  competitive ratio of SRPT is at least $\Omega(m^{1-1/\alpha})$.

\item For the special case where all jobs have unit size, we show that
  round robin (RR) routing is optimal, and the best known competitive
  ratio results on speed scaling to minimize the flow-time plus energy
  in the single server setting apply in the multiple server setting as
  well.

\item We also consider the stochastic setting, where jobs arrive
  according to a Poisson process with i.i.d. sizes. This case turns
  out to be significantly easier than the worst case; we show that
  with $P(s) = s^{\alpha}$ ($\alpha > 1$), random routing and a simple
  gated-static speed scaling algorithm achieves a constant competitive
  ratio, e.g., $2$ for $P(s) = s^2.$

\end{itemize}

\section{System Model}

Let the input consist of $J$ jobs, where job $j$ arrives (is released)
at time $a_j$ and has work/size $w_j$.  There are $m$ homogenous
servers, each with the same power function $P(s),$ where $P(s)$
denotes the power consumed while running at speed $s$. 
Any job can be processed by any of the $m$ servers.

The speed $s$ is the rate at which work is executed by any of the
server, and $w$ amount of work is completed in time $w/s$ by any
server if run at speed $s$ throughout time $w/s$.  A job $j$ is
defined to be complete at time $c_j$ if $w_j$ amount of work has been
completed for it, possibly by different servers.  We assume that
preemption is allowed, i.e., a job can be suspended and later
restarted from the point at which it was suspended. Moreover, we also
assume that job migration is allowed, i.e., if a job is preempted it
can be processed later at a different server than the one from which
it was preempted. Thus, a job can be processed by different servers at
different intervals, but at any given time it can be processed by only
server, i.e., no job splitting is allowed. The flow time $f_j$ for job
$j$ is $f_j = c_j-a_j$ (completion time minus the arrival time) and
the overall flow time is $F= \sum_{j} f_j$. From here on we refer to
$F$ as just the flow time.  Note that $F = \int n(t) dt$, where $n(t)$
is the number of unfinished jobs at time $t$. Thus, flow time can also
be interpreted as the \emph{cumulative holding cost}, where
instantaneous holding cost at time $t$ equals $n(t).$

Let server $k$ run at speed $s_k(t)$ at time $t$.  The energy cost is
defined as $\sum_{k=1}^mP(s_k(t))$ summed over the flow time.
Choosing larger speeds reduces the flow time, however, increases the
energy cost, and the natural objective function that has been
considered extensively in the literature is the sum of flow time and
energy cost \revision{(see, for example,
  \cite{bansal2007speed,Andrew2010,greiner2009bell,lam2012improved})},
which we define as\footnote{It is also natural to take the objective to be a linear 
    combination  of flow time and energy, i.e., $\int n(t) dt + \beta 
    \int  \sum_{k=1}^m P(s_k(t))dt,$ where $\beta > 0$ weighs the energy 
    cost relative to the delay cost. However, note that since the 
    factor~$\beta$ may be absorbed into the power function, we will work with 
    the objective~\eqref{eq:cost} without loss of generality.}
\begin{equation}\label{eq:cost}
  C =  \int n(t) dt + \int  \sum_{k=1}^m P(s_k(t))dt.%
\end{equation}

Any online algorithm only has causal information, i.e., it becomes
aware of job $j$ only at time $a_j$. Any online algorithm with
multiple servers has to make two causal decisions: routing; that
specifies the assignment of jobs to servers, and scheduling; that
specifies a job to be processed by each server and at what speed at
each time. Let the cost \eqref{eq:cost} of an online algorithm $A$ be
$C_A$.  Moreover, let the cost of \eqref{eq:cost} for an offline
optimal algorithm that knows the job arrival sequence $\sigma$ (both
$a_j$ and $w_j$) in advance be $C_{\opt}$. Then the competitive ratio
of the online algorithm $A$ for $\sigma$ is defined as
\begin{equation}\label{eq:compratio}
  \revision{\sfc_A(\sigma) = \max_{\sigma} \frac{C_A(\sigma)}{C_{\opt}(\sigma)},}
\end{equation}
and the objective function considered in this paper is to find an
online algorithm that minimizes the worst case competitive ratio
\begin{equation}\label{eq:obj}
\sfc^\star = \min_A   \sfc_A(\sigma).
\end{equation}

We will also consider stochastic input $\sigma$ where both $a_j$ and
$w_j$ are chosen stochastically, in which case our definition for
competitive ratio for $A$ will be
\begin{equation}\label{eq:compratiostoc}
\sfc_A = \frac{\Exp{C_A}}{\Exp{C_{\opt}}},
\end{equation} 
where the expectation is with respect to the stochastic input; see
Section~\ref{sec:stoc} for the details. Correspondingly, the goal is
to come up with an online algorithm that minimizes $\sfc_A.$

In Sections \ref{sec:CR_upper} to \ref{sec:CR_lower}, we study the
worst-case setting, and present the results for the stochastic setting
in Section~\ref{sec:stoc}.

\section{Worst Case Competitive Ratio: Upper Bounds}
\label{sec:CR_upper}

In this section, we present our results on constant competitive
policies for scheduling and speed scaling in a multi-server
environment. We propose an online policy that performs SRPT scheduling,
where the instantaneous speed of each server is a function of the
number of outstanding jobs in the system. We prove that this policy is
constant competitive for a broad class of power
functions. Specifically, \revision{our upper bounds on the competitive
  ratio} depend only on the power function, but not on the number of
jobs, their sizes, or the number of servers.

\subsection{SRPT Algorithm}
In this section, we consider the SRPT algorithm for \revision{scheduling and routing}, and
analyze its competitive ratio when the server speeds are chosen as
follows. Let $n(t)$ and $n_o(t)$ denote the number of unfinished jobs
with the SRPT algorithm and $\opt$ (the offline optimal algorithm)
respectively, at time $t.$ Moreover, let $A(t)$ and $O(t)$ be the set
of active jobs with the SRPT algorithm and $\opt,$
respectively. Recall that the SRPT algorithm maintains a single queue
and serves the $\min\{m, n(t)\}$ shortest jobs at any time
$t.$

The speed for job $k \in A(t)$ with the SRPT algorithm is chosen as
\begin{equation}\label{eq:speeddef}
s_k(t) =
\begin{cases}
 &   P^{-1}\left(\frac{n(t)}{m}\right)\  \text{if}\  n(t) \ge m, \\
&  P^{-1}(1), \ \ \ \ \ \ \ \  \ \text{otherwise}.
\end{cases}
\end{equation}
\revision{Idle servers, i.e., servers not processing any job, are run
  at speed~0.} The above speed scaling rule can be interpreted as
follows. Under \eqref{eq:speeddef},
$\sum_{k \in A(t)} P(s_k(t)) = n(t),$ i.e., \emph{the instantaneous
  power consumption is matched to the instantaneous job holding cost.}
\revision{This type of \emph{energy-proportional} speed scaling has be
  shown to be near-optimal in the single server setting
  \cite{Andrew2010}, where the basic motivation for this choice of
  speed is \emph{local optimality}---it is the optimal speed choice at
  time $t$ assuming no future job arrivals after time $t$.}

Our main result (Theorem~\ref{thm:srptimproved}) is proved under the
following assumptions on the power function.
\begin{assumption}
  \label{ass:Pconvex}
  $P:\R_+ \ra \R_+$ with $P(0) = 0$ is differentiable, strictly increasing, and
  strictly convex function, which implies $\lim_{s \ra \infty}P(s) =
  \infty,$ and $\bar{s}:=\inf\{s > 0 \ |\ P(s) > s \} < \infty$.
\end{assumption}
\revision{
\begin{assumption}
  \label{ass:new}
  For $x,y > 0,$ $P(xy) \leq P(x) P(y).$
\end{assumption}
Assumptions~1 and~2 are satisfied by power functions of the form
$P(s) = c s^{\alpha},$ where $\alpha > 1$ and
$c \geq 1.$}\footnote{\revision{Our main result (see
  Theorem~\ref{thm:srptimproved}) can also be proved under the
  following assumption, that is weaker than Assumption~\ref{ass:new}:
  There exists $\gamma \geq 1$ such that $P(xy) \leq \gamma P(x) P(y)$
  for all $x,y > 0.$ In this case, the upper bound on the competitive
  ratio would depend on $\gamma.$}}
%
We are now ready to state our main result, which shows that our SRPT
algorithm is constant competitive.
\begin{theorem}\label{thm:srptimproved} 
  Under \revision{Assumptions~\ref{ass:Pconvex} and~\ref{ass:new}},
  the SRPT algorithm with speed scaling \eqref{eq:speeddef} is
  $c$-competitive, where
  $$c = P(2-1/m) \left(2 + \frac{2}{P^{-1}(1)} \max(1,P(\bar{s}))\right).$$
\end{theorem}
Taking $P(s) = s^{\alpha}$ for $\alpha >1,$ the competitive ratio
equals $4(2-1/m)^{\alpha}.$ To prove Theorem \ref{thm:srptimproved},
we use a potential function argument, where the potential function is
defined as follows.  Let $n_o(t,q)$ and $n(t,q)$ denote the number of
unfinished jobs under $\opt$ and the algorithm, respectively, with
remaining size at least $q$. In particular, $n_o(t,0) = n_o(t)$ and
$n(t,0) = n(t)$. Let
$$d(t,q) = \max\left\{0, \frac{n(t,q) - n_o(t,q)}{m}\right\}.$$ 
\revision{
We consider the potential function
\begin{equation}\label{defn:phi}
\Phi(t) = \Phi_1(t)+\Phi_2(t),
\end{equation}
where $\Phi_1(\cdot)$ and $\Phi_2(\cdot)$ are defined as follows.}
\begin{equation*}\label{defn:phi1} 
  \Phi_1(t) = c_1\int_{0}^\infty f\left(d(t,q)\right) dq,
\end{equation*}
\revision{where $c_1$ is a positive constant whose value will be
  specified later. The function $f:\{i/m:\ i \in \Z_+\} \ra \R_+$ is
  defined as follows: $f(0) = 0,$ and for $i \geq 1,$ $$f(\frac{i}{m})
  = f(\frac{i-1}{m}) + \Delta(\frac{i}{m}),$$ where $\Delta(x) :=
  P'(P^{-1}(x)).$}
\begin{equation*}\label{defn:phi2}
\Phi_2(t) = c_2 \int_{0}^\infty (n(t,q) - n_o(t,q)) dq,
\end{equation*}
\revision{where $c_2$ is another positive constant whose value will be
  specified later.}

The $\Phi_1(t)$ part of the potential function is a multi-server
generalization of the potential function
in~\cite{bansal2009speedconf}, while the $\Phi_2(t)$ part is novel.
\revision{Specifically, $\Phi_2(t)$ ensures that the potential
  function has the requisite `drift' (see the
  condition~\eqref{eq:mothereq} below) at times when the algorithm has
  fewer than $m$ servers active.}
  
Let the speed of job $k \in O(t)$ under $\opt$ at time~$t$
  be ${\tilde s}_k(t)$. Suppose we can show that for any input
  sequence $\sigma,$
\begin{equation}\label{eq:mothereq} 
  n(t) + \sum_{k \in A(t)} P(s_k(t)) + \frac{d\Phi(t)}{dt} \le
  c\biggl(n_o(t) + \sum_{k\in O(t)} P({\tilde s}_k(t))\biggr)
\end{equation} 
almost everywhere \revision{(in $t,$ with respect to the Lebesgue
  measure)}, where $c > 0$ does not depend of $\sigma,$ and that
$\Phi(t)$ satisfies the following {\it boundary} conditions (proved in
Proposition \ref{prop:boundary}; see Appendix~\ref{app:boundary}):
\begin{enumerate}
\item Before any job arrives and after all jobs are finished,
  $\Phi(t)= 0$, and
\item $\Phi(t)$ does not have a positive jump discontinuity at any
  point of non-differentiability.
\end{enumerate} 
Then, integrating \eqref{eq:mothereq} with respect to $t$, we get that
\revision{
\begin{align*}
  \int \biggl(n(t) &+ \sum_{k\in A(t)}
  P(s_k(t)) \biggr) dt \\
  &\le \int c\biggl(n_o(t) + \sum_{k \in O(t)} P({\tilde s}_k(t))\biggr) dt,
\end{align*}} 
which is equivalent to showing that
$C_{SRPT}(\sigma) \le c \ C_{\opt}(\sigma)$ for any input
$\sigma$ as required.

The intuition for the form of the competitive ratio in
Theorem~\ref{thm:srptimproved} is as follows.
\revision{\begin{definition}[Speed Augmentation] Consider the problem of minimizing flow time 
$m$ servers that have fixed speed. With speed augmentation, an online algorithm is allowed to use $m$ servers with fixed speed $\beta>1$, while the $\opt$ is restricted to use $m$ servers with fixed speed $1$. 
\end{definition}
The following result summarizes the power of speed augmentation with the SRPT algorithm for minimizing the flow time
\begin{lemma}\label{lem:optsrpt} \cite{phillips1997optimal} 
For flow time minimization problem with $m$ servers with fixed speed, the flow time of SRPT algorithm with speed augmentation of $\beta=(2-1/m)$ is at most the flow time of the $\opt$ with $m$ speed-$1$ servers.
\end{lemma}
\begin{rem} In this paper, we do not consider speed augmentation, but
  speed scaling, where both the online algorithm and the $\opt$ have
  access to identical servers, but whose speeds can be varied with an
  associated energy cost. In Lemma \ref{lem:OS}, we show how we can
  use the speed augmentation result (Lemma~\ref{lem:optsrpt}) for
  analyzing the speed scaling problem.
\end{rem}

\begin{definition}\label{defn:OS} Consider an SRPT-restricted optimal
  offline algorithm for minimizing the flow time plus energy with
  variable speed servers \eqref{eq:cost}, that we call $\opt$-$SRPT$
  ($OS$).  $OS$ has non-causal knowledge of the complete job arrival
  sequence $\sigma$, but is required to follow the SRPT discipline at
  each time $t$, i.e., schedule the $m$ jobs with the shortest
  remaining sizes.  Given the SRPT restriction, $OS$ can choose speeds
  at each time that minimizes the flow time plus energy
  \eqref{eq:cost} knowing the future job arrival sequence.
\end{definition}

Note that we neither know the optimal speed choice of the $\opt$ nor
of the $OS$ algorithm.
\begin{lemma}\label{lem:OS} Algorithm $OS$ is $P(2-1/m)$-competitive
  with respect to $\opt$ for minimizing \eqref{eq:cost}.
\end{lemma}
\begin{proof}
Let the speed at time $t$ chosen by $\opt$ be $s_{\opt}(t)$.
The claim follows from Lemma~\ref{lem:optsrpt}, since $OS$ can choose speed $(2-1/m)s_\opt(t)$ at time $t$, and get exactly the same flow-time as the $\opt$. Increasing speed for $OS$ only results in an extra multiplicative energy cost of $P(2-1/m)$ over $\opt$.  
Thus, the flow time plus energy cost of $OS$ is at most $P(2-1/m)$ times the flow time plus energy cost of the $\opt$.
\end{proof}

For proving Theorem \ref{thm:srptimproved}, we will show that  the SRPT algorithm with speed scaling \eqref{eq:speeddef} is $c$ competitive with respect to $OS$ for some $c$, which implies that  the SRPT algorithm with speed scaling \eqref{eq:speeddef} is $c P(2-1/m)$ competitive with respect to $\opt$ itself.
}


For smaller values of $\alpha,$ the result of Theorem
\ref{thm:srptimproved} can be further improved for the special case of
power-law power functions, as described in the next theorem.
\begin{theorem}\label{thm:srpt} With $P(s) = s^\alpha$ and for any
  $\alpha \in (1,2)$, the SRPT-based algorithm with speed scaling
  \eqref{eq:speeddef} is $c$-competitive, where
  $$c = 3 + \frac{2}{2-\alpha}.$$
\end{theorem}
The proof of Theorem~\ref{thm:srpt} is similar in spirit to that of
Theorem~\ref{thm:srptimproved}, but without assuming that $\opt$
follows SRPT. It also uses the same potential function $\Phi$ (see
\eqref{defn:phi}), and directly tries to bound the increase in $\Phi$
because of processing of the jobs by the algorithm and $\opt$. The
limitation on $\alpha$ appears because without enforcing that $\opt$
follows SRPT, we cannot apply a technical lemma (Lemma
\ref{lem:bansal}) jointly on the change made to $\Phi$ by the
algorithm and the $\opt$, but individually. \revision{While the
  competitive ratio stated in Theorem~\ref{thm:srpt} grows unboundedly
  as $\alpha \ua 2$, it is less than the competitive ratio established
  in Theorem~\ref{thm:srptimproved} when $\alpha$ is close to 1. This
  is because our argument for Theorem~\ref{thm:srpt} does not enforce
  $\opt$ to follow SRPT, thereby saving on the penalty of
  $(2-1/m)^{\alpha}.$} The proof of Theorem~\ref{thm:srpt} is provided
in Appendix~\ref{app:srpt_2}, while the remainder of this section is
devoted to the proof of Theorem~\ref{thm:srptimproved}.

\begin{proof}[Proof of Theorem~\ref{thm:srptimproved}]
  \revision{To prove the theorem we will show that  the SRPT algorithm with speed scaling \eqref{eq:speeddef} is $c$-competitive with respect to 
  the $OS$ algorithm (Definition \ref{defn:OS}), which together with Lemma \ref{lem:OS} implies that  the SRPT algorithm with speed scaling \eqref{eq:speeddef} is $cP(2-1/m)$ competitive with respect to $\opt$. For the rest of the proof, subscript $o$ refers to quantities for $OS$ algorithm.} 

In the following, we show that \eqref{eq:mothereq} is true for a
suitable choice of $c.$ To show \eqref{eq:mothereq}, we bound
$d\Phi/dt$ via individually bounding $d\Phi_1/dt$ and $d\Phi_2/dt$ in
Lemmas~\ref{lemma:phi1_OPT-SRPT} and~\ref{lemma:phi2_OPT-SRPT}
below. Note that it suffices to show that \eqref{eq:mothereq} holds at
any instant~$t$ which is not an arrival or departure instant under the
algorithm or $OS$. For the remainder of this proof, consider any
such time instant $t.$ For ease of exposition, we drop the index $t$
from $n(t,q), n(t,q_o), n(t), n_o(t), s_k(t)$ and ${\tilde s}_k(t),$
since only a fixed (though generic) time instant $t$ is under
consideration.

\begin{lemma}
  \label{lemma:phi1_OPT-SRPT}
  For $n \geq m$,
  \begin{align*}
    d\Phi_1/dt \le  &c_1 n_o - c_1 n + c_1 \left(\frac{m-1}{2}\right) + c_1 \sum_{k \in O} P(\tilde{s}_k),
  \end{align*} 
  while for $n < m,$
  \begin{align*}
    d\Phi_1/dt &\le c_1 n_o - c_1 \frac{n(n+1)}{2m} + c_1 \sum_{k \in O} P(\tilde{s}_k)
  \end{align*}
\end{lemma}

\begin{lemma}
  \label{lemma:phi2_OPT-SRPT}
  $d\Phi_2/dt \leq -c_2 \min(m,n) P^{-1}(1)$
  \begin{align*}
    \qquad \qquad  +  c_2\sum_{k \in O} \max\{P(\bar{s}), P({\tilde s}_k)\}
  \end{align*}
\end{lemma}
\revision{\begin{rem}\label{rem:drift} The basic idea behind the
    expressions derived in Lemmas~\ref{lemma:phi1_OPT-SRPT} and
    \ref{lemma:phi2_OPT-SRPT} is to upper bound the drift or the
    derivative of potential function (which counts the difference of
    work remaining with the proposed algorithm and the $OS$
    algorithm). Both the proposed algorithm and the $OS$ algorithm are
    working on their respective $m$ shortest jobs, and thereby
    decreasing and increasing the potential function, respectively.
    For the specific speed choice, \eqref{eq:speeddef}, we show that
    the rate of change of the potential function is at most a
    (negative) constant times the number of the remaining jobs with
    the proposed algorithm plus a (positive) constant times the
    holding cost of the $OS$ algorithm. Since the instantaneous
    holding cost for the algorithm is $n(t) + P(s(t)) = 2 n(t)$ from
    \eqref{eq:speeddef}, this bound on the drift, as we show next, is
    sufficient to show \eqref{eq:mothereq} for some~$c$.
\end{rem}}

Using Lemmas~\ref{lemma:phi1_OPT-SRPT} and~\ref{lemma:phi2_OPT-SRPT}
(proved in Appendix~\ref{app:lemmas_main_thm}), we now prove
\eqref{eq:mothereq} by considering the following two cases:

\noindent [{\bf Case 1: $n \geq m.$}] $n + \sum_{k \in A} P(s_k) + d\Phi(t)/dt$
\begin{align*}
  \stackrel{(a)}\le & n + n + c_1 n_o - c_1 n + c_1
  \left(\frac{m-1}{2}\right) + c_1\sum_{k \in O} P({\tilde s}_k)
  \\ &\quad -c_2 m P^{-1}(1) + c_2\sum_{k \in O} \max\{P(\bar{s}),
  P({\tilde s}_k)\} \\ \leq& (c_1+c_2) \sum_{k \in O} P({\tilde s}_k)
  + (c_1 + c_2P(\bar{s})) n_o + n(2-c_1) \\ & \quad +
  \left[c_1\left(\frac{m-1}{2}\right) -c_2mP^{-1}(1) \right]
  \\ \stackrel{(b)}\le & \bigl(c_1 + c_2 \max(1,P(\bar{s}))\bigr)
  \bigl( n_o + \sum_{k \in O} P({\tilde s}_k) \bigr)
\end{align*}
Here, $(a)$ follows from Lemmas~\ref{lemma:phi1_OPT-SRPT} and
\ref{lemma:phi2_OPT-SRPT}, and since $P(s_k) = n/m$ when $n \geq m$
(see \eqref{eq:speeddef}), while $(b)$ follows by setting $c_1 = 2$
and $c_2 \geq 1/P^{-1}(1).$

\noindent [{\bf Case 2: $n < m.$}]
$n + \sum_{k \in A} P(s_k) + d\Phi(t)/dt$
\begin{align*}
  \stackrel{(a)}\le & n + n + c_1 n_o - c_1 \frac{n(n+1)}{2m} + c_1
  \sum_{k \in O} P(\tilde{s}_k) -c_2 n P^{-1}(1) \\ &\quad +
  c_2\sum_{k \in O} \max\{P(\bar{s}), P({\tilde s}_k)\} \\ \leq&
  \bigl(c_1 + c_2 \max(1,P(\bar{s}))\bigr) \bigl( n_o + \sum_{k \in O} P({\tilde s}_k) \bigr) + n(2-c_2
  P^{-1}(1))\\ \stackrel{(b)}\le & \bigl(c_1 + c_2
  \max(1,P(\bar{s}))\bigr) \bigl( n_o + \sum_{k \in O} P({\tilde s}_k)
  \bigr)
\end{align*}
Once again, $(a)$ follows from Lemmas~\ref{lemma:phi1_OPT-SRPT} and
\ref{lemma:phi2_OPT-SRPT}, and since $P(s_k) = 1$ when $n < m$ (see
\eqref{eq:speeddef}), while $(b)$ follows by setting $c_2 =
2/P^{-1}(1).$

This proves \eqref{eq:mothereq} for $$c = c_1 + c_2 \max(1,P(\bar{s}))
= \left(2 + \frac{2}{P^{-1}(1)} \max(1,P(\bar{s}))\right).$$
\end{proof}
In the next section, we consider a special case when all jobs have
unit size, but their arrival instants are still worst case, for which
we can improve the competitive ratio guarantees.

\subsection{Equal Sized Jobs}

Assume that all jobs have equal size, which is taken to be~$1$ without
loss of generality. There are $m$ servers and jobs are assigned on
arrival to one of the $m$ servers for service. $\opt$ refers to the
offline optimal policy.
We propose the following policy ${\cal U}$. Each job on its arrival is
assigned to servers in a round-robin fashion, and each server $k$ uses
speed $s_k(t) = P^{-1}(n_k(t))$, where $n_k(t)$ is the number of
unfinished jobs that have been assigned to server~$k$.
\begin{theorem} 
  \label{thm:equalsize}
With unit job sizes, under Assumption~\ref{ass:Pconvex}, ${\cal U}$ is
$2$-competitive.
\end{theorem}
\begin{proof} In Proposition \ref{prop:rr_centralQ}, we show that when
  all jobs are of unit size, $\opt$ follows round robin
  scheduling. Thus, ${\cal U}$ and $\opt$ see the same set of arrivals
  on each server. The result follows from~\cite{Andrew2010}, which
  shows that choosing speed $s_k(t) = P^{-1}(n_k(t))$ for a single
  server system is $2$-competitive.
\end{proof}
\begin{prop} 
  \label{prop:rr_centralQ}
  With unit job sizes, under Assumption~\ref{ass:Pconvex}, $\opt$
  performs round robin dispatch across servers.
\end{prop}

\begin{proof}
  Let us assume that $\opt$ can hold arriving jobs in a central queue
  before dispatch to one of the $m$ servers. It suffices to show that
  even in this expanded space of policies, $\opt$ can be assumed to
  perform round robin dispatch without loss of optimality (WLO).
  \begin{enumerate}
  \item From the convexity of the power function, it follows that
    $\opt$ serves each job at a constant speed. Labeling jobs in the
    order of their arrival, let $s_j$ denote the speed at which
    job~$j$ is served.
  \item \revision{Once in service, a job is not interrupted until it
      completes. (Indeed, it can be shown via a simple interchange
      argument that it is not beneficial to interrupt a job in service
      to begin service on another, when all jobs are of the same size
      to begin with.)}
  \item WLO, we may assume that $\opt$ dispatches jobs for service in a
    FCFS manner.
  \end{enumerate}

  {\bf Claim~1:} WLO, $\opt$ completes jobs in the order of their
  arrival.

  It follows from Claim~1 that $\opt$ can be assumed to perform round
  robin WLO. \revision{Indeed, Claim~1 implies that job completions
    from servers would occur in the same order as the order of
    dispatch, allowing for a cyclic (round robin) dispatch across
    servers.}
  
  {\bf Proof of Claim 1:} Let $a_l$ denote the time when job~$l$
  begins service and let $d_l$ denote the time when the same job
  completes service. Suppose the claim does not hold, i.e., there
  exist $i,j$ where $i < j$ such that $d_j < d_i.$ We now demonstrate
  an alternative power allocation that is strictly better for $\opt$.

  Note that $a_i \leq a_j.$ Let $r \leq 1$ denote the remaining work
  of job~$i$ at time $a_j.$ Clearly, $d_j < d_i$ implies that
  $s_j > s_i.$ Fix $\delta \in (0,1/s_j]$ such that
  \begin{equation}
    \label{eq:rr_1}
    s_j \delta + s_i \left(\frac{1}{s_j}-\delta\right) = r.
  \end{equation}
  Consider the following power allocation:
  \begin{enumerate}
  \item Starting at time $a_j,$ job~$i$ is served at speed $s_j$ for
    $\delta$ time units, and at speed $s_i$ for $\frac{1}{s_j}-\delta$
    time units
  \item Starting at time $a_j,$ job~$j$ is served at speed $s_j$ for
    $\frac{1}{s_j}-\delta$ time units, and at speed $s_i$ for
    $\frac{r}{s_i} - \frac{1}{s_j}+\delta$ time units
  \end{enumerate}

  From~\eqref{eq:rr_1}, it is not hard to see that under this new
  power allocation, the departure instants of jobs~$i$ and $j$ are
  interchanged, i.e., job~$i$ completes at time $d_j,$ whereas job~$j$
  completes at time $d_i.$ Moreover, under the above power allocation,
  the cost of $\opt$ remains unchanged. Indeed, the increase in the delay
  cost of job~$j$ is exactly compensated by the decrease in the delay
  cost of job~$i.$ Moreover, the energy cost remains unchanged, and
  the cost associated with all remaining jobs remains unchanged as
  well (we simply interchange all subsequent dispatches between the
  servers serving jobs~$i$ and~$j$).

  Now, from the convexity of the power function, it follows that we
  can strictly decrease the energy cost of $\opt$ by running jobs~$i$ and
  $j$ at constant speeds from time $a_j$, such that the completion
  times remain unchanged.

  This gives us a contradiction, and completes the proof of the claim.
\end{proof}


\ignore{
This result also extends to the case where the power function is of
the form $$\hat{P}(s) = \left\{
  \begin{array}{cl}
    0 & \text{ for }s = 0 \\
    p_0 + P(s) & \text{ for } s > 0
  \end{array}
\right.,$$ where $p_0 > 0$ and $P(s)$ is convex and continuous with
$P(0) = 0.$ The main difference is that in the latter model, $\opt$ might
idle, i.e., keep jobs in the central queue even when servers are
idle. But this does not affect the proof of
Proposition~\ref{prop:rr_centralQ}.
}

\section{Worst Case Competitive Ratio: Lower bounds}
\label{sec:CR_lower}

In the previous section, we showed that while SRPT scheduling is not
constant-competitive in a multi-server environment without speed
scaling, it can be made constant-competitive when speed scaling is
allowed. However, one issue with implementing SRPT on multiple servers
is the need for job migration. In this section, we show that a broad
class of greedy non-migratory policies is not constant-competitive.

\revision{For deriving lower bounds, we consider $P(s) = s^\alpha$ in the rest of the section for brevity.} 
We begin by stating the following preliminary result.
\begin{lemma}
  \label{lemma:single_server_LB}
  On a single server, consider a single burst of $n$ jobs, with sizes
  $x_n \leq x_{n-1} \leq \cdots \leq x_1.$ The cost incurred by $\opt$ in
  processing this burst equals $c \sum_{k = 1}^n x_k k^{1-1/\alpha},$
  where the constant $c$ depends on $\alpha.$
\end{lemma}
The proof of Lemma \ref{lemma:single_server_LB} follows by direct
computation of the optimal speeds for each job that minimize the flow
time plus energy cost \eqref{eq:cost}.

\subsection{Greedy algorithms}

\revision{The main idea in deriving lower bounds on a class of
  non-migratory online algorithms is to construct an arrival
  sequence~$\sigma$ for which that class of algorithms exhibits poor
  performance because of the inability to migrate jobs between servers
  once assigned (given that those assignments were made based on only
  causal information about $\sigma$). In contrast, $\opt$ can have
  much lower cost than the online algorithm, even without migration,
  by exploiting its non-causal information about $\sigma.$}
 
\begin{lemma}
  \label{lemma:workload_greedy}
  Consider the class of policies that routes an incoming job to a
  server with the least amount of unfinished workload. All policies
  in this class have a competitive ratio that is $\Omega(m^{1-1/\alpha}).$
\end{lemma}

\begin{proof}
  Consider the following instance: A burst of $m-1$ jobs, each having
  size $w$ arrives at time $0$, and another burst of $w$ jobs, each
  having size 1 arrives at time $0^+.$

  Any workload-based greedy policy would assign the first $m-1$ jobs
  of size $w$ to $m-1$ different servers, and the $w$ jobs of size 1
  to the remaining server. By Lemma~\ref{lemma:single_server_LB}, the
  cost incurred by any such algorithm is at-least \revision{(for some
    $c$ depending only on $\alpha$)}
  $$c(m-1)w + c \sum_{k = 1}^w k^{1-1/\alpha} \geq c(m-1)w + c' w^{2-1/\alpha},$$
  for some constant $c'$.

  Consider now an algorithm $A$ that assigns the first $m-1$ jobs of
  size $w$ to $m-1$ different servers and then distributes the $w$
  jobs of size 1 uniformly among all $m$ servers. The algorithm $A$
  then performs scheduling and speed scaling on each server as per
  single server $\opt$. The cost incurred by $A$ (which upper bounds the
  cost under $\opt$) equals (using Lemma~\ref{lemma:single_server_LB})
  \begin{align*}
    &c \sum_{k = 1}^{w/m}
      k^{1-1/\alpha} + c(m-1)\left[ w + \sum_{k = 2}^{w/m+1} k^{1-1/\alpha} \right] \\
    & \leq c'' m \left(\frac{w}{m}\right)^{2-1/\alpha} + c m w
  \end{align*}
  
  Now, setting $w = m^d$ for large enough $d,$ \revision{$( d > 1/(1-1/\alpha))$} we see that the
  competitive ratio of any workload-based greedy policy is
  $\Omega(m^{1-1/\alpha}).$
\end{proof}

It follows from the proof of Lemma~\ref{lemma:workload_greedy} that
the competitive ratio of any policy that routes an incoming job to
a server that has been assigned the least aggregate workload so far
(including completed as well as queued workload) is also
$\Omega(m^{1-1/\alpha}).$

\begin{lemma}
  \label{lemma:jsq_greedy}
  Consider the class of policies that route an incoming job to a
  server with the least number of queued jobs (join the shortest
  queue (JSQ)). All policies in this class have a competitive ratio that is
  $\Omega(m^{1-1/\alpha}).$
\end{lemma}

\begin{proof}
  Consider the following instance: $m^2$ jobs arrive in quick
  succession, causing any JSQ-based policy to perform round robin
  routing. Every $m$th arriving job has size $w,$ while all remaining
  jobs have size 1.

  Thus, under any JSQ-based policy, one server would get $m$ jobs of
  size $w$ routed to it, whereas all other servers would get $m$
  jobs of size~1. Thus, the cost under any such policy is at least $cw  \sum_{k = 1}^m k^{1-1/\alpha} + c (m-1) \sum_{k = 1}^m k^{1-1/\alpha} 
  \geq c' w m^{2-1/\alpha} + c' (m-1)m^{2-1/\alpha}$.

  Consider an algorithm $A$ that routes the jobs uniformly across
  the servers, such that each server gets $m-1$ jobs of size 1, and
  one job of size $w.$ Post routing, $A$ performs scheduling and
  speed scaling on each server as per single server $\opt$. The cost
  incurred by $A$ (which upper bounds the cost of $\opt$) is
  thus
  $m\left[cw + c \sum_{k = 1}^{m-1} k^{1-1/\alpha} \right] \leq c''
  \left[mw + m^{3-1/\alpha} \right]$.

  Now, setting $w = m^d$ for large enough $d,$ we see that the
  competitive ratio of any JSQ-based policy is 
  $\Omega(m^{1-1/\alpha}).$
\end{proof}

It is also clear from the above proof that any policy that performs
round robin routing would have a competitive ratio that is
$\Omega(m^{1-1/\alpha}).$

\subsection{SRPT-based algorithms}

In this section, we consider the following class of non-migratory
SRPT-based policies: Let $y_j(t)$ denote the least remaining
processing time among all jobs queued at server~$j.$ If server~$j$ is
idle at time $t$, then set $y_j(t) = 0.$ Consider now a job of size
$x$ arriving into the system at time $t.$ If the set
$\{j:\ y_j(t) > x\}$ is non-empty, then the job is assigned to a
server from this set. Else, the job is assigned to any server, or held
in a central queue. Each server may preempt between jobs queued at
that server. But jobs once assigned to a certain server must complete
service at that server, i.e., migration is not allowed.

\begin{lemma}
  \label{lemma:srpt_NonMigration}
  Consider the class of non-migratory SRPT-based policies described
  above. All policies in this class have a competitive ratio that is
  $\Omega(m^{1-1/\alpha}).$
\end{lemma}

\begin{proof}
  Consider the following instance: $m-1$ jobs of size 1 arrive at time
  0, and $m$ jobs of sizes $w,w-\epsilon,\cdots,w-(m-1)\epsilon$ arrive
  in quick succession right after. 

  Any non-migratory SRPT-based policy would route the $m-1$ jobs of
  unit size to $m-1$ different servers, and the next $m$ jobs to the
  remaining server. Thus, the cost incurred is at
  least $(w-m\epsilon) c \sum_{k = 1}^{m} k^{1-1/\alpha} \geq
  c'(w-m\epsilon) m^{2-1/\alpha}$.

  Consider next a policy $A$ that routes the first $m-1$ unit
  sized jobs to $m-1$ different servers, and distributes the next $m$
  jobs across all servers. Post routing, $A$ performs scheduling and
  speed scaling on each server as per single server $\opt$. The cost
  incurred by $A$ (which upper bounds the cost of $\opt$) is thus at
  most $m c \left(w + 2^{1-1/\alpha} \right)$. Now, setting $w =
  m^d$ for large enough $d,$ we see that the competitive ratio of any
  non-migratory SRPT-based policy is $\Omega(m^{1-1/\alpha}).$
\end{proof}


\section{Stochastic Input}\label{sec:stoc}

In this section, we consider a stochastic model for the job
arrivals. Jobs arrive according to a Poisson process of rate
$\lambda,$ and have i.i.d. sizes. Let $X$ denote a generic job
size. We assume that $\Exp{X} < \infty.$ The \emph{load}, which is the
rate at which work is submitted to the system, is given by $\Lambda =
\lambda \Exp{X}.$

The performance metric under consideration is the stationary variant
of the flow time plus energy metric considered for the worst-case
analysis, i.e.,
\begin{equation}
  \label{eq:stochastic_metric}
  C = \Exp{T} + \Exp{E},
\end{equation}
where $T$ denotes the steady state response time, and $E$ denotes the
energy required to serve a job in steady state.\footnote{Of course,
  for this metric to be meaningful, we restrict attention to policies
  that are regenerative, and thus have a meaningful steady state
  behavior. We also note that it is straightforward to extend the
  results of this section to a metric that is a linear combination of
  $\Exp{T}$ and $\Exp{E}.$} In the present section, we restrict
attention to power functions of the form $P(s) = s^{\alpha},$ where
$\alpha > 1.$


In the following, we generalise a result proved in
\cite{wierman2009power} for the single server setting to the
multi-server setting. Specifically, we show that a policy that routes
each job randomly, and runs each server at a constant speed
$s^*(\Lambda)$ when active, is constant competitive. Note that the
speed chosen depends on the load $\Lambda,$ which needs to be known or
learnt. Policies of this type are referred to in
\cite{wierman2009power} as \emph{gated static} policies.

Specifically, the proposed algorithm $\mathcal{S}$ is the following:
Arriving jobs are routed to any server uniformly at random. Each server
performs processor sharing (PS) scheduling using a fixed speed
$s^*(\Lambda),$ which is the optimal static speed to minimize the
metric \eqref{eq:stochastic_metric} on that (single) server.

We begin our analysis by deriving a lower bound on the performance of
any routing and speed scaling policy.

\subsection{Lower Bound}


Let $s_i$ denote the time-averaged speed of server $i.$ We have
\begin{align}
  \lambda C & \geq \lambda \Exp{E} = \sum_{i=1}^m \Exp{P(s_i)}
  \nn \\ & \stackrel{(a)}\geq \sum_{i=1}^m P(\Exp{s_i})  \stackrel{(b)}\geq \sum_{i=1}^m P(\Lambda/m)
  \nn \\ &=\frac{\Lambda^{\alpha}}{ m^{\alpha-1}}, \label{eq:stoclb1}
  \end{align}
where $(a)$ follows from  an application of the Jensen's inequality,
while $(b)$ follows from the convexity of the power function $P$, given
that $\sum_{i=1}^m \Exp{s_i} = \Lambda$ (for stability).

Next, we derive an alternate lower bound on $C$. Consider the case
when only a single job of size $X$ arrives. This job is run at a
constant speed $s^\star$ that minimizes its response time plus energy
consumption, i.e., $s^\star = \arg \inf_{s>0} \frac{X}{s} +
\frac{X}{s} P(s) = \left(\frac{1}{\alpha-1}\right)^{1/\alpha}.$
This yields the following lower bound on the performance of any
algorithm 
\begin{equation}
  \label{eq:stoclb2}
  \lambda C \ge \Lambda \alpha
  (\alpha-1)^{\left(\frac{1}{\alpha}-1\right)}.
\end{equation}

Combining \eqref{eq:stoclb1} and \eqref{eq:stoclb2} gives us 
\begin{equation}
  \label{eq:stoclb3}
  \lambda C \ge \max\left( \Lambda \alpha
  (\alpha-1)^{\left(\frac{1}{\alpha}-1\right)},
  \frac{\Lambda^{\alpha}}{ m^{\alpha-1}} \right).
\end{equation}

Next, we characterize the performance of the proposed policy and bound
its competitive ratio.

\subsection{Performance under policy $\mathcal{S}$}
\label{subsec:stocperf}

Under random routing, each server sees a Poisson arrival process with
rate $\lambda/m.$ Thus, the performance under
metric~\eqref{eq:stochastic_metric} when operating each server at
speed $s$ when active with PS scheduling is given by
\begin{equation}
  \label{eq:stocub}
C(s) = \frac{\Exp{X}}{(s-\Lambda/m)} + \frac{\Exp{X}}{s} P(s).
\end{equation} 
Choosing $s^*(\Lambda) = \argmin_{s > 0} C(s),$ we get that the cost \eqref{eq:stochastic_metric} of the
algorithm $\mathcal{S}$ is given by $C(s^*(\Lambda)).$

\ignore{ Hence, the expected load seen by any server is
  $\Lambda/m$. Thus, the cost of any one server when run at speed $s$
  whenever there is at least one unfinished job is
\begin{equation}\label{eq:stocub} c(s) = \frac{\Lambda/m}{(s-\Lambda/m)} + \frac{\Lambda}{ms} P(s).
\end{equation} If $P(s) = s^\alpha$, then $s^\star = \arg \min_{s>0} c(s)$ satisfies $ (\alpha-1)(s^\star)^{\alpha-2} (s^\star-\Lambda/m)^2=1$. The overall cost is $C \le mc(s^\star)$.
}

\begin{theorem}
  In the stochastic input setting, the competitive ratio of the
  algorithm $\mathcal{S}$ \revision{is upper bounded by} a constant
  that depends on $\alpha$ but not on $\lambda,$ the job size
  distribution, or $m$.
\end{theorem}
\begin{rem}\label{rem:stoc} The basic idea of the proof is to compare the two lower bounds on the cost \eqref{eq:stochastic_metric} derived in \eqref{eq:stoclb1} and \eqref{eq:stoclb2},  with the upper bound on the cost \eqref{eq:stochastic_metric} of $\cS$,  the random routing and gated static policy, derived in \eqref{eq:stocub}.
 \end{rem}

\begin{proof}
  The proof follows by comparing the performance $\mathcal{S}$ with
  the lower bound \eqref{eq:stoclb3} that holds for any algorithm.

  Indeed, for any algorithm $A,$
  \begin{align*}
    \frac{C_{\mathcal{S}}}{C_A} &\leq \frac{\lambda
      c(s^*(\Lambda))}{\max\left( \Lambda \alpha
      (\alpha-1)^{\left(\frac{1}{\alpha}-1\right)},
      \frac{\Lambda^{\alpha}}{ m^{\alpha-1}} \right)} \\ &\leq
    \frac{\lambda c(1+\Lambda/m)}{\max\left( \Lambda \alpha
      (\alpha-1)^{\left(\frac{1}{\alpha}-1\right)},
      \frac{\Lambda^{\alpha}}{ m^{\alpha-1}} \right)}
    \\ &=\frac{\Lambda + \Lambda
      \left(1+\Lambda/m\right)^{\alpha-1}}{\max\left( \Lambda \alpha
      (\alpha-1)^{\left(\frac{1}{\alpha}-1\right)},
      \frac{\Lambda^{\alpha}}{ m^{\alpha-1}} \right)} \\ & \leq
    \frac{\Lambda + \Lambda
      \left(1+\Lambda/m\right)^{\alpha-1}}{\min(1,\alpha
      (\alpha-1)^{\left(\frac{1}{\alpha}-1\right)}) \max\left(
      \Lambda, \frac{\Lambda^{\alpha}}{ m^{\alpha-1}} \right)} \\ &
    \leq \frac{1+2^{\alpha-1}}{\min(1,\alpha
      (\alpha-1)^{\left(\frac{1}{\alpha}-1\right)})}.
  \end{align*}
\end{proof}

The above bound can be tightened for the case $\alpha = 2,$ since
$s^*(\Lambda)$ can be computed explicitly in this case.

\begin{corollary}
  For $\alpha=2$, in the stochastic input setting, the competitive
  ratio of the algorithm $\mathcal{S}$ is at most $2$.
\end{corollary}
\begin{proof}
  For $\alpha =2$, from \eqref{eq:stocub}, we get
  $s^{*}(\Lambda) = 1 + \frac{\Lambda}{m},$ and thus, the performance
  $C_{\mathcal{S}}$ under the algorithm satisfies
  $$\lambda C_{\mathcal{S}} = \Lambda^2/m + 2\Lambda.$$ Now, from
  \eqref{eq:stoclb3}, $\lambda C_A \ge \max\{\Lambda^2/m, 2\Lambda\}$
  under any algorithm $A$, which implies the statement of the
  corollary.
\end{proof}

\section{Concluding Remarks}

In this paper, we show that SRPT can be made constant competitive in
the multi-server speed scaling environment with respect to the flow
time plus energy metric. This presents an interesting contrast to the
case when server speeds are constant, where it is known that SRPT has
an unbounded competitive ratio with respect to the flow time
metric. We also show that the multi-server speed scaling problem is
easy in the absence of job size variability; simple round robin
dispatch in conjunction with a single-server speed scaling rule is
near-optimal. Finally, we show that a broad class of policies based on
greedy non-migratory dispatch rules do not admit a constant
competitive ratio.

In contrast, in the stochastic setting, we show that random routing,
along with a gated static speed setting is constant
competitive. However, the required speed is a function of the load,
which needs to be learnt.

While SRPT is a well studied scheduling policy in the multiple server
setting, one issue with implementing SRPT in practice is the need for
migration. Considering that there is a cost associated with migration
of a job across servers in practice, a natural generalization would be
to include this cost of migration in the performance metric. How to
optimally tradeoff flow time, energy consumption, and migration costs
is an interesting open problem for the future. However, it is easy to
bound the performance of the SRPT-based speed scaling algorithm
proposed in this paper accounting for migration costs. Indeed, in a
job sequence consisting of $J$ jobs, SRPT performs at most $J$
migrations. Thus, assuming a fixed cost of each migration, our
SRPT-based algorithm remains constant competitive with respect to the
flow time plus energy plus migration cost metric if one assumes a
lower bound on the size of each job; in this case, the migration cost
is at most a constant factor of the flow time.

Finally, we note that while there is a considerable literature on
speed scaling in parallel multi-server environments, we are not aware
of any work on speed scaling in tandem queueing systems, and more
generally, on a queueing network. Coming up with constant competitive
speed scaling algorithms in these settings is an interesting avenue
for future work.

\appendices



\section{Proof of Proposition \ref{prop:boundary}}
\label{app:boundary}
\begin{proposition}\label{prop:boundary}
  $\Phi(\cdot)$ as defined in \eqref{defn:phi} satisfies boundary
  conditions \revision{1) and 2)}.
\end{proposition}
\begin{proof} Note that Condition~1 is satisfied; before any job is
  released and after all jobs are finished, $\Phi(t)= 0,$
  \revision{since $n(t,q)= n_o(t,q)=0$ for all $q.$}
  \revision{Whenever a new job arrives/is released, $n(t,q)- n_o(t,q)$
    and $d(t,q)$ do not change for any~$q$,} so $\Phi$ remains
  unchanged.  Similarly, whenever a job is completed by the algorithm
  or $\opt$, $d(t,q)$ or $n(t,q) -n_o(t,q)$ is changed for only a
  single point of $q=0,$ which does not introduce a discontinuity in
  $\Phi(t).$ Thus, Condition~2 is also satisfied.
\end{proof}

\section{Proof of Lemmas~\ref{lemma:phi1_OPT-SRPT} and
\ref{lemma:phi2_OPT-SRPT}}
\label{app:lemmas_main_thm}

To prove Lemmas~\ref{lemma:phi1_OPT-SRPT} and
\ref{lemma:phi2_OPT-SRPT} we need the following technical lemma from
\cite{bansal2009speedconf}.

\begin{lemma}\label{lem:bansal}[Lemma 3.1 in~\cite{bansal2009speedconf}]
  For $s_k, {\tilde s}_k, x \ge 0$,
\begin{align*}
\Delta(x)(-s_k + {\tilde s}_k) \le& \left(-s_k +P^{-1}(x)\right)\Delta(x) 
 + P({\tilde s}_k) -x.
\end{align*}
\end{lemma}

\revision{We now give the proofs of Lemmas~\ref{lemma:phi1_OPT-SRPT}
  and \ref{lemma:phi2_OPT-SRPT}.}
\begin{proof}[Proof of Lemma~\ref{lemma:phi1_OPT-SRPT}]
  Let $q(i)$ and $q_o(i)$ denote, respectively, the size of
  the~$i^{th}$ shortest job \emph{in service} under the algorithm and
  $OS.$

  \noindent {\bf Case 1: $n \geq m.$} Suppose that $OS$ is serving
  $r$ jobs, where $r \leq m.$ Define $\tilde{n}(q) = \max(n(q),n-r)$,
  and $\tilde{n}_o(q) = \max(n_o(q),n_o - r)$. \revision{We assume for
    now that all jobs under the algorithm and $OS$ have distinct
    remaining sizes. We later point out how our argument generalizes
    when this assumption is relaxed. Now, given the preceding
    assumption,} the function $g(q) := \tilde{n}(q) - \tilde{n}_o(q)$
  satisfies the following properties.
  \begin{enumerate}
  \item $g(0) = n-n_o,$ $g(q) \ra n-n_0$ as $q \ra \infty,$
  \item $g$ is piecewise constant and left-continuous, with a downward
    jump of 1 at $q = q(i),$ $1 \leq i \leq r,$ and an upward jump of
    1 at $q = q_o(i),$ $1 \leq i \leq r.$
  \end{enumerate}
  
  Consider the change in $\Phi_1$ due to $OS$ ($n_o(q) \rightarrow
  n_o(q)-1$ for $q=q_o(1), \cdots, q_o(r)$):
  \begin{align}
    d\Phi_1 =& c_1 \sum_{i = 1}^r
    \biggl[f\left(\frac{n(q_o(i))-n_o(q_o(i))+1}{m}\right) \nn \\
    &\quad -
    f\left(\frac{n(q_o(i))-n_o(q_o(i))}{m}\right) \biggr]\tilde{s}_i dt \nn \\
    \stackrel{(a)}= &c_1 \sum_{i = 1}^r \Delta\left(\frac{n(q_o(i)) -
        n_o(q_o(i)) + 1}{m}\right) \tilde{s}_i dt \nn \\
    \stackrel{(b)}\le & c_1 \sum_{i = 1}^r
    \Delta\left(\frac{\tilde{n}(q_o(i)) - \tilde{n}_o(q_o(i)) +
        1}{m}\right) \tilde{s}_i dt \nn \\
    =& c_1 \sum_{i = 1}^r \Delta\left(\frac{g(q_o(i))+ 1}{m}\right)
    \tilde{s}_i dt \label{eq:dphi1_opt_new_case1}
\end{align}
In writing $(a)$ we take $\Delta(i/m) = 0$ for $i \leq 0.$~$(b)$ holds
since $\tilde{n}_o(q_o(i)) = n_o(q(i))$ for $1 \leq i \leq r,$ and
$\tilde{n}(q) \geq n(q)\ \forall\ q.$

Next, consider the change in $\Phi_1$ due to the algorithm ($n(q)
\rightarrow n(q)-1$ for $q=q(1), \cdots, q(m)$):

\begin{align}
  d\Phi_1 =& c_1 \sum_{i = 1}^m
  \biggl[f\left(\frac{n(q(i))-1-n_o(q(i))}{m}\right) \nn \\
  &\quad -
  f\left(\frac{n(q(i))-n_o(q(i))}{m}\right) \biggr] s_i dt \nn \\
  =& -c_1 \sum_{i = 1}^m \Delta\left(\frac{n(q(i)) -
      n_o(q(i))}{m}\right) s_i dt \nn \\
  \stackrel{(a)}\le & -c_1 \sum_{i = 1}^r
  \Delta\left(\frac{\tilde{n}(q(i)) - \tilde{n}_o(q(i))}{m}\right) s_i dt \nn\\
  & \quad -c_1 \sum_{i = r+1}^m \Delta\left(\frac{n(q(i)) -
      n_o(q(i))}{m}\right)
  s_i dt \nn \\
  \stackrel{(b)}\le & -c_1 \sum_{i = 1}^r
  \Delta\left(\frac{g(q(i))}{m}\right) s_i dt \nn \\
  &\quad -c_1 \sum_{i = r+1}^m \Delta\left(\frac{n-i+1-n_o}{m}\right)
  s_i dt \label{eq:dphi1_alg_new_case1}
\end{align}
Here, $(a)$ holds because $n(q(i)) = \tilde{n}(q(i))$ for
$1 \leq i \leq r,$ and $\tilde{n}_o(q) \geq n_o(q)$ for all $q.$~$(b)$
follows since $n(q(i)) \geq n-i+1.$\footnote{$n(q(i)) = n-i+1$ if the
  algorithm has exactly one job with remaining size $q(i).$ If
  multiple jobs have the same remaining size $q(i)$ under the
  algorithm, then we have $n(q(i)) \geq n-i+1.$}
 
We now combine \eqref{eq:dphi1_opt_new_case1} and
\eqref{eq:dphi1_alg_new_case1} to capture the overall change in
$\Phi_1.$ In doing so, we make the following crucial observation.

\noindent {\bf Claim~1:} For each $i \in \{1,2,\cdots,r\},$
\revision{there exists} a unique $j \in \{1,2,\cdots,r\}$ such that
$g(q(i)) \geq g(q_o(j))+ 1.$

To see that this claim is true, note that at each job with remaining
size $q(k)$ ($1 \leq k \leq r$) being served by the algorithm
contributes a down-tick of magnitude~1 in $g$ at $q(k).$ Similarly,
each job with remaining size $q_o(k)$ ($1 \leq k \leq r$) being served
by $OS$ contributes an up-tick of magnitude~1 in $g$ at
$q_o(k).$\footnote{The magnitude of the discontinuity in $g$ at $q$
  thus equals
  $|\{j \in R:\ q_o(j) = q\}| - |\{j \in R:\ q(j) = q\}|$, where
  $R = \{1,2,\cdots,r\}.$} It is therefore clear that each down-tick
from $l$ to $l-1$ can be mapped to an unique up-tick from $l-1$ to
$l.$ Moreover, at the downtick, say at $q(i),$ we have
$g(q(i)) \geq l$ (because $g$ is left-continuous), and at the
corresponding up-tick, say at $q_o(j),$ we have $g(q_o(j)) \leq l-1,$
implying $g(q_o(j)) +1 \leq l,$ (again, because $g$ is
left-continuous). This proves the claim.

Based on the above observation, combining
\eqref{eq:dphi1_opt_new_case1} and \eqref{eq:dphi1_alg_new_case1}, we
can now bound the overall change in $\Phi_1$ as
follows.\footnote{\revision{Note that until now, we took $\tilde{s}_i$
    to be the speed of the $i$th shortest job under $OS.$ However,
    once we match drift terms between the algorithm and $OS$, the
    indices $i$ can get shuffled. To maintain notational simplicity,
    we do not make this re-labelling explicit. The indices $i$ will be
    understood to span the active jobs/servers under the
    algorithm/$OS.$}}
\begin{align}
d\Phi_1/dt \leq& c_1 \sum_{i = 1}^r \Delta\left(\frac{g(q(i))}{m}\right) (-s_i
+ \tilde{s}_i) \nn \\ & \quad -c_1 \sum_{i = r+1}^m \Delta\left(\frac{n-i+1-n_o}{m}\right)
s_i.  \label{eq:phi1_case1_1}
\end{align}
Invoking Lemma~\ref{lem:bansal}, terms in the first summation
of \eqref{eq:phi1_case1_1} can be bounded as
 $$\Delta\left(\frac{g(q(i))}{m}\right) (-s_i + \tilde{s}_i) \leq
 P(\tilde{s}_i) - \frac{g(q(i))}{m},$$ since $s_k =
 P^{-1}\left(\frac{n}{m}\right) \geq
 P^{-1}\left(\frac{g(q(i))}{m} \right)$. Lemma~\ref{lem:bansal} can
 also be used to bound the terms of the second summation
 of \eqref{eq:phi1_case1_1} as
 $$- \Delta\left(\frac{n-i+1-n_o}{m} \right) s_i \leq - \frac{n-i+1-n_o}{m}$$
 (taking $\tilde{s}_k$ in the statement of Lemma~\ref{lem:bansal} to
 be zero). Combining the above bounds, we arrive at $   d\Phi_1/dt$
 \begin{align*}
 \leq& c_1 \sum_{i = 1}^r \left[ P(\tilde{s}_i) -
                    \frac{g(q(i))}{m}\right]  - c_1 \sum_{i = r+1}^m
                    \left(\frac{n-i+1-n_o}{m}\right).
 \end{align*}
 Finally, noting that $g(q(i)) \geq n - i +1 - n_o$ for
 $1 \leq i \leq r,$ we conclude that 
   \begin{align*}
     d\Phi_1/dt \leq& c_1 \sum_{i = 1}^r P(\tilde{s}_i) - c_1 \sum_{i = 1}^m
                      \left(\frac{n-i+1-n_o}{m}\right), \\
     =& c_1 n_0 - c_1 n + c_1 \left(\frac{m-1}{2}\right) + c_1 \sum_{i = 1}^r P(\tilde{s}_i).
 \end{align*}

 \revision{This proves the bound on $d\Phi_1/dt$ claimed in the
   statement of Lemma~\ref{lemma:phi1_OPT-SRPT} for the case
   $n \geq m.$}

\revision{
 \begin{rem}
   We now remark how the preceding argument generalizes if multiple
   jobs under the algorithm and/or $OS$ have identical remaining
   size. Suppose the algorithm has $k$ active jobs with remaining size
   $\hat{q},$ and $OS$ has $\ell$ active jobs, also with remaining
   size $\hat{q}.$ If these $k+\ell$ jobs are not being operated at
   the same speed, we simply disregard the change in $\Phi_1$ at that
   instant of time; after all, it suffices to prove
   \eqref{eq:mothereq} almost everywhere in time. On the other hand,
   if all $k+\ell$ jobs are being served at the same speed, say
   $\hat{s},$ the change in $\Phi_1$ is captured as follows.

   Suppose that $k \geq \ell.$ Note that $g$ has a downward jump of
   $k-\ell$ at $\hat{q}.$ The change in $\Phi_1$ due to the $k+\ell$
   jobs with remaining size $\hat{q}$ in service is given
   by:
   \begin{align*}
     d\Phi_1 &= c_1
     \biggl[f\left(\frac{n(\hat{q})-n_o(\hat{q}) - (k-\ell)}{m}\right) 
     - f\left(\frac{n(\hat{q})-n_o(\hat{q})}{m}\right) \biggr] \hat{s} dt  \\
     &= -c_1 \sum_{j = 0}^{k - \ell-1} \Delta\left(\frac{n(\hat{q}) -
       n_o(\hat{q}) - j}{m}\right) \hat{s} dt \\
     &= -c_1 \sum_{j = 0}^{k - \ell-1} \Delta\left(\frac{n(\hat{q}) -
       n_o(\hat{q}) - j}{m}\right) \hat{s} dt \\
     &\qquad + c_1 \sum_{j = k-\ell}^{k - 1} \Delta\left(\frac{n(\hat{q}) -
            n_o(\hat{q}) - j}{m}\right) (-\hat{s} + \hat{s}) dt
   \end{align*}
   The $k-\ell$ terms in the first sum above (which correspond to the
   downward jump of $k-\ell$ in $g$ at $\hat{q}$) can be matched to
   $k-\ell$ up-ticks in $g$ due to $OS$ as in Claim~1. The $\ell$
   terms in the second sum make zero contribution to the change in
   $\Phi_1,$ but can still be bounded using Lemma~\ref{lem:bansal}, so
   that the steps after \eqref{eq:phi1_case1_1} can be followed to
   arrive at the final bound on $d\Phi_1/dt.$

   A similar argument can be made if $k < \ell.$
 \end{rem}
 }

 \noindent {\bf Case~2: $n < m.$} Let $r$ denote the number of jobs in
 service under $OS.$ Define $h(q):= n(q) - n_o(q).$ \revision{As in
   Case~1, we assume for simplicity of exposition that all jobs under
   the algorithm and $OS$ have distinct remaining sizes. The general
   case can be handled as was pointed out in the proof of Case~1.} The
 rate of change of $\Phi_1$ can be expressed as follows:

 \begin{align*}
   \frac{d\Phi_1}{dt} &= c_1 \sum_{i = 1}^r
   \Delta\left(\frac{h(q_o(i))+ 1}{m}\right) \tilde{s}_i - c_1 \sum_{i
     = 1}^n \Delta\left(\frac{h(q(i))}{m}\right) s_i \\
   & \leq c_1 \sum_{i = 1}^{n_o} \Delta\left(\frac{h(q_o(i))+
       1}{m}\right) \tilde{s}_i - c_1 \sum_{i = 1}^n
   \Delta\left(\frac{h(q(i))}{m}\right) s_i
 \end{align*}

 \noindent {\bf Claim~2:} For each $i \in \{1,2,\cdots,n_o\}$ such
 that $h(q_o(i)) \geq 0,$ one can find a unique
 $j \in \{1,2,\cdots,n\}$ such that $h(q(j)) \geq h(q_o(i))+1.$

The proof of the above claim follows along the same lines as the proof
  of Claim~1 for $n \geq m.$ Note that $h$ is piecewise constant and
  left-continuous, with $h(0) = n-n_o,$ $h(q) = 0$ for large enough
  $q,$ has upward jumps at $q_o(i)$ $(i \leq n_o)$ at downward jumps
  at $q(i)$ $(i \leq n).$ Thus, any uptick in $h(\cdot)$ from $l-1$ to
  $l$ for $l \geq 1$ can be mapped to a unique downtick from $l$ to
  $l-1.$ The rest of the argument is identical to that of Claim~1.

 Based on the above observation, suppose that a subset $J$ of
 algorithm terms are matched with $OS$ terms.

 \begin{align*}
   \frac{d\Phi_1}{dt} & \leq c_1 \sum_{i \in J}
   \Delta\left(\frac{h(q(i))}{m}\right) (-s_i + \tilde{s}_i) \\ & \quad + c_i
   \sum_{i \notin J} \Delta\left(\frac{h(q(i))}{m}\right) (-s_i)
 \end{align*}

Applying Lemma~\ref{lem:bansal} as before, 
\begin{align*}
  \frac{d\Phi_1}{dt} & \leq c_1 \left(\sum_{i \in J} P(\tilde{s}_i) -
    \frac{h(q(i))}{m}\right) - c_1 \sum_{i \notin J} \frac{h(q(i))}{m}  \\
  &\leq  c_1 \sum_{i \in O} P(\tilde{s}_i) - \sum_{i = 1}^n
  \frac{h(q(i))}{m} \\
  &\leq c_1 \sum_{i \in O} P(\tilde{s}_i) - c_1 \sum_{i = 1}^n \frac{n - i +
    1 - n_o}{m} \\
  &\leq c_1 \sum_{i \in O} P(\tilde{s}_i) + c_1 n_o - c_1 \frac{n(n+1)}{2m}
\end{align*}
\end{proof}

\begin{proof}[Proof of Lemma~\ref{lemma:phi2_OPT-SRPT}]
The rate of change in $\Phi_2$ is    
\begin{align*}
 d\Phi_2/dt &= - c_2\sum_{k \in A} s_k  + c_2 \sum_{k \in O} \tilde{s}_k   \\
         & \leq - c_2 \min(n,m) P^{-1}(1) + c_2 \sum_{k \in O} \max(P(\bar{s}), P(\tilde{s}_k))
\end{align*}
The bounding of the first term above uses $s_k \geq P^{-1}(1).$ The
bounding of the second term is based on: (i) $\tilde{s}_k \leq
P(\bar{s})$ when $\tilde{s}_k \leq \bar{s},$ and (ii)
$\tilde{s}_k \leq P(\tilde{s}_k)$ when $\tilde{s}_k > \bar{s}.$
\end{proof}

\section{Proof of Theorem~\ref{thm:srpt}}
\label{app:srpt_2}

Unlike in the proof of Theorem~\ref{thm:srptimproved}, we now make no
assumptions on the scheduling of $\opt.$ We use the same potential
function $\Phi$ as before (see \eqref{defn:phi}), and show
that~\eqref{eq:mothereq} holds for a suitable $c.$ Note that it
suffices to show that~\eqref{eq:mothereq} holds at any instant~$t$
which is not an arrival or departure instant under the algorithm or
$\opt.$ For the remainder of this proof, consider any such time
instant $t.$ For ease of exposition, we drop the index $t$ from
$n(t,q), n(t,q_o), n(t), n_o(t), s_k(t)$ and ${\tilde s}_k(t),$ since
only a fixed (but generic) time instant $t$ is under consideration.

Our proof is based on the following lemmas.
\begin{lemma}\label{lem:phi1}
  For $n \geq m$, \begin{align*} d\Phi_1/dt \le &c_1 n_o -
  c_1(2-\alpha) n + c_1(2-\alpha)\left(\frac{m-1}{2}\right)\\ &\qquad
  + c_1\sum_{k \in O} P({\tilde s}_k), \end{align*} while for $n <
  m,$ \begin{align*} d\Phi_1/dt &\le c_1 n_o
  + \frac{c_1(2-\alpha)n}{2} + c_1\sum_{k \in O} P({\tilde
  s}_k).  \end{align*}
\end{lemma}

\begin{lemma}
    \label{lem:phi2}
    $d\Phi_2/dt \leq -c_2 \min(m,n) + c_2\sum_{k \in O} P({\tilde s}_k)$
\end{lemma}

Using Lemmas~\ref{lem:phi1} and~\ref{lem:phi2}, we now prove
\eqref{eq:mothereq} by considering the following two cases:

\noindent {\bf Case 1: $n \geq m.$} 
 
\begin{align*}
  &n + \sum_{k \in A} P(s_k) + d\Phi(t)/dt \\ \leq & n +
  n + c_1 n_o - c_1(2-\alpha) n +
  c_1(2-\alpha)\left(\frac{m-1}{2}\right)\\ &\quad + c_1\sum_{k \in O} P({\tilde s}_k) + -c_2
  m  +
  c_2\sum_{k \in O} P({\tilde s}_k) \\ \leq &(c_1 +
  c_2) \bigl( n_o + \sum_{k \in O} P({\tilde s}_k) \bigr) +
  n[2-c_1(2-\alpha)] \\ & \quad + \left[c_1(2-\alpha)\left(\frac{m-1}{2}\right)
  -c_2m \right] \\ \stackrel{(a)}\le & (c_1 + c_2) \bigl( n_o
  + \sum_{k \in O} P({\tilde s}_k) \bigr).
\end{align*}
Here, $(a)$ follows by setting $c_1 = \frac{2}{\revision{2-\alpha}},$ and
$c_2 \geq 1.$

\noindent {\bf Case 2: $n < m.$} 

\begin{align*}
  &n + \sum_{k \in A} P(s_k) + d\Phi(t)/dt \\
  \stackrel{(a)}\le & 2n + c_1 n_o + \frac{c_1(2-\alpha)n}{2} + c_1\sum_{k \in O} P({\tilde s}_k) 
  -c_2 n \\ & \quad +  c_2\sum_{k \in O} P({\tilde s}_k) \\
  \leq & (c_1 + c_2) \bigl( n_o + \sum_{k \in O} P({\tilde s}_k) \bigr) + n(2 + c_1(2-\alpha)/2 - c_2) \\
  \stackrel{(a)}\le & (c_1 + c_2) \bigl( n_o + \sum_{k =1}^{|O|} P({\tilde s}_k) \bigr),
\end{align*}
Here, $(a)$ follows setting $c_1 = \frac{2}{\revision{2-\alpha}},$ and
$c_2 = 3.$

This proves \eqref{eq:mothereq} for $c = c_1 + c_2 = 3
+ \frac{2}{2-\alpha}.$ It now remains to prove Lemmas~\ref{lem:phi1}
and~\ref{lem:phi2}.

\begin{proof}[Proof of Lemma~\ref{lem:phi1}]
  Let $q(i)$ denote the size of the~$i^{th}$ shortest job \emph{in
    service} under the algorithm. Note that since the algorithm
  performs STPT scheduling, $q(i)$ is also the size of the $i^{th}$
  shortest job in the system under the algorithm. Let $q_o(i)$ denote
  the size of the~$i^{th}$ largest job \emph{in service} under $\opt.$
  
  \noindent {\bf Case~1: $n \geq m$}. When $n \geq m,$ since the
  algorithm processes the $m$ shortest jobs, the change in $\Phi_1$
  because of the algorithm ($n(q) \rightarrow n(q)-1$ for
  $q=q(1), \cdots, q(m)$) is
  \begin{align}
    d\Phi_1 &= c_1\sum_{k=1}^m \biggl[f\left(\frac{n(q(k)) - 1- n_o(q(k))}{m}\right)s_k dt \nonumber\\
            &\qquad - c_1f\left(\frac{n(q(k))
              - n_o(q(k))}{m}\right)s_k dt\biggr] \nonumber \\ 
            &\stackrel{(a)}\le - c_1\sum_{k=1}^m\Delta\left(\frac{n(q(k))- n_o(q(k))}{m}\right)s_k dt \label{eq:dphi1_alg_1} \\ 
            &\stackrel{(b)}\le -c_1\sum_{k=1}^m\Delta\left(\frac{n-k+1 - n_o}{m}\right)s_k dt \label{eq:dphi1_alg_2}
  \end{align}  
  In writing $(a)$ we take $\Delta(i/m) = 0$ for $i \leq 0.$~$(b)$
  follows since $n(q(k)) \geq n-k+1$, and $n_o(q(k)) \le n_o$ for all
  $k.$ Next, we bound the terms of \eqref{eq:dphi1_alg_2}
  using~Lemma \ref{lem:bansal}. For those terms where the argument of
  $\Delta(\cdot)$ is non-negative, Lemma~\ref{lem:bansal} implies that

  \begin{equation}
    \label{eq:dummy456}
    -\Delta\left(\frac{n-k+1 - n_o}{m}\right)s_k \leq - \left(\frac{n- k+1 -n_o}{m}\right);
  \end{equation}
  take $\tilde{s}_k = 0$ and note that
  $$s_k = P^{-1}\left(\frac{n}{m}\right) >
  P^{-1}\left(\frac{n-k+1-n_o}{m}\right).$$ \revision{The
    bound~\eqref{eq:dummy456} also holds for those terms where the
    argument of $\Delta(\cdot)$ is negative; in this case, the left
    hand side of the inequality is zero, and the right hand side is
    positive.} This yields the bound \begin{equation} d\Phi_1 \leq
    -c_1\sum_{k=1}^m\left(\frac{n- k+1 -n_o}{m}\right)
    dt.  \label{eq:dphi1_alg} \end{equation}

  We now consider the change in $\Phi_1$ due to $\opt.$
  \begin{align}
    d\Phi_1 &=  c_1\sum_{k \in O}   \Delta\left(\frac{n(q_o(k)) + 1- n_o(q_o(k))}{m}\right){\tilde s}_k(t) dt \nn \\
    \label{eq:dphi1_opt1}
& \stackrel{(a)}\le c_1\sum_{k \in O} \Delta\left(\frac{n +1 - k}{m}\right){\tilde s}_k(t) dt, 
\end{align}
where $(a)$ follows since $n(q_o(k)) \leq n,$ and
$n_o(q_o(k)) \geq k.$\footnote{The reader should verify that
  the bound on $n_o(q_o(i))$ applies even if multiple active jobs
  under $\opt$ have identical sizes.} Applying Lemma \ref{lem:bansal}
with $s_k=0,$ we get that the change in $\Phi_1$ because of $\opt$ is
\begin{align*}
   d\Phi_1/dt & \le c_1 (\alpha-1) \sum_{k=1}^{|O|}\left(\frac{n + 1 - k}{m}\right)
  +  c_1\sum_{k  \in O} P({\tilde s}_k),
\end{align*}
since \revision{$P^{-1}(i)P^{'}(P^{-1}(i))- i = (\alpha-1) i.$}
Finally, since $|O| \le m,$ we have that the change in $\Phi_1$
because of $\opt$ satisfies
\begin{equation}
\label{eq:dphi1_opt}
d\Phi_1/dt \leq  c_1 (\alpha-1) \sum_{k=1}^{m}\left(\frac{n- k+1}{m}\right) +  c_1\sum_{k \in O} P({\tilde s}_k).
\end{equation}

Combining \eqref{eq:dphi1_alg} and \eqref{eq:dphi1_opt}, the overall
change in $\Phi_1$ satisfies
\begin{align*}\nn
  d\Phi_1/dt \le  &c_1 n_o - c_1(2-\alpha) n +  c_1(2-\alpha)\left(\frac{m-1}{2}\right) \nonumber \\
  &\qquad + c_1\sum_{k \in O} P({\tilde s}_k).
\end{align*} 

\noindent {\bf Case~2: $n < m.$} Our approach in capturing the change
in $\Phi_1$ due to the algorithm is the same as that in Case~1, except
that the summations in \eqref{eq:dphi1_alg_1} and
\eqref{eq:dphi1_alg_2} only run from $k = 1$ to $k = n.$ An
application of Lemma~\ref{lem:bansal} as before implies that the 
component of $d\Phi_1/dt$ because of the algorithm satisfies
\begin{align}\label{eq:dphi1_alg_case2}
  d\Phi_1 \leq  -c_1\sum_{k=1}^n\left(\frac{n- k+1 -n_o}{m}\right) dt.
\end{align} 

The analysis of the impact of $\opt$ on $\Phi_1$ also proceeds as in
Case~1, except that the summation in \eqref{eq:dphi1_opt1} only runs
from $k = 1$ to $k = n;$ note that the remaining terms in the sum are
zero. An application of Lemma~\ref{lem:bansal} as before implies that
the component of $d\Phi_1/dt$ because of $\opt$ satisfies 
\begin{equation}
\label{eq:dphi1_opt_case2}
d\Phi_1/dt \leq  c_1 (\alpha-1) \sum_{k=1}^{n}\left(\frac{n- k+1}{m}\right) +  c_1\sum_{k \in O} P({\tilde s}_k)
\end{equation}

Combining \eqref{eq:dphi1_alg_case2} and \eqref{eq:dphi1_opt_case2},
the overall change in $\Phi_1$ is bounded as
\begin{align*}
  d\Phi_1/dt &\le c_1 n_o + \frac{c_1(2-\alpha)n}{2} + c_1\sum_{k \in O} P({\tilde s}_k).
\end{align*}
\end{proof}

\begin{proof}[Proof of Lemma~\ref{lem:phi2}]
The proof is similar to that of Lemma~\ref{lemma:phi2_OPT-SRPT},
except that we exploit the specific form of the power function.
\begin{align*}
 d\Phi_2/dt &= - c_2\sum_{k \in A} s_k  + c_2 \sum_{k \in O} \tilde{s}_k   \\
         & \leq - c_2 \min(n,m) + c_2 \sum_{k \in O} P(\tilde{s}_k)
\end{align*}
In the above bound, we use the fact that for $P(s) = s^{\alpha},$ the
minimum speed utilized by any algorithm is $P^{-1}(1) = 1.$ Thus,
$s_k,\tilde{s_k} \geq 1,$ and $\tilde{s_k} \leq P(\tilde{s_k}).$
\end{proof}


\bibliographystyle{IEEEtran}
\bibliography{refs}

\end{document}